%% file: main.tex
\documentclass[letterpaper, 10 pt, conference]{ieeeconf}  

\IEEEoverridecommandlockouts                    

\usepackage{amsmath}
\usepackage{amssymb}
\usepackage{algorithm}
\usepackage[noend]{algpseudocode}
\usepackage{color}
\usepackage{graphicx}
\usepackage{booktabs}
\usepackage{enumerate}
\usepackage{setspace}
\usepackage[hidelinks]{hyperref}

\usepackage{amsthm}

\newtheorem{thm}{Theorem}[section]

\newtheorem{defn}{Definition}[section]

\newtheorem{rem}{Remark}

\title{\LARGE \bf Data-Efficient Control Barrier Function Refinement} 

\author{Bolun Dai$^{1}$, Heming Huang$^{1}$, Prashanth Krishnamurthy$^{1}$, Farshad Khorrami$^{1}$
\thanks{$^{1}$Control/Robotics Research Laboratory, Electrical~\&~Computer Engineering Department, Tandon School of Engineering, New York University, Brooklyn, NY, 11201
{\tt\small bd1555@nyu.edu, hh2564@nyu.edu, prashanth.krishnamurthy@nyu.edu, khorrami@nyu.edu}
}
}

\newcommand{\x}{\mathbf{x}}

\begin{document}

\maketitle
\thispagestyle{empty}
\pagestyle{empty}

\maketitle
\input{abstract}
\input{intro}
\input{preliminary}
\input{formulation}
\input{method}
\input{simulation}
\input{conclusion}

\vspace*{-0.05in}
\bibliographystyle{IEEEtran}
\bibliography{IEEEabrv, refs.bib}
\end{document}

%% file: abstract.tex
\begin{abstract}
Control barrier functions (CBFs) have been widely used for synthesizing controllers in safety-critical applications. When used as a safety filter, it provides a simple and computationally efficient way to obtain safe controls from a possibly unsafe performance controller. Despite its conceptual simplicity, constructing a valid CBF is well known to be challenging, especially for high-relative degree systems under nonconvex constraints. Recently, work has been done to learn a valid CBF from data based on a handcrafted CBF (HCBF). Even though the HCBF gives a good initialization point, it still requires a large amount of data to train the CBF network. In this work, we propose a new method to learn more efficiently from the collected data through a novel prioritized data sampling strategy. A priority score is computed from the loss value of each data point. Then, a probability distribution based on the priority score of the data points is used to sample data and update the learned CBF. Using our proposed approach, we can learn a valid CBF that recovers a larger portion of the true safe set using a smaller amount of data. The effectiveness of our method is demonstrated in simulation on a unicycle and a two-link arm.
\end{abstract}

%% file: intro.tex
\section{Introduction}
With the rapid integration of automated systems in our daily lives, the ability to ensure operation safety for automatic systems has gained more attention~\cite{DaiKPK21}~\cite{DaiSKK21}. In optimal control, the safety requirements can be formulated as optimization constraints. Following this line of thought, one of the most popular approaches is model predictive control (MPC). However, the main obstacle to directly applying MPC is its solution time. MPC is more commonly used in a multi-layered control architecture, where MPC generates a reference trajectory and a reactive control tracks the reference trajectory at a higher frequency. One issue with this approach is that the constraint is only active when possible unsafe scenarios might occur within the preview horizon. When the time horizon is small, this leads to abrupt control actions, which is not ideal for most safety-critical applications. Increasing the time horizon is also not ideal since the solution time would also increase.

Recently, CBF-based methods~\cite{DBLP:conf/eucc/AmesCENST19} have become increasingly popular in the field of safety-critical control. CBFs act as a certificate of the safety of a state with respect to a safe set. For zeroing CBFs, a state is safe when the CBF at that state is positive and unsafe when it is negative. Additionally, CBFs can be used to alter unsafe control actions into safe control actions in a minimally invasive way. Since CBFs are commonly used within a quadratic program (QP)~\cite{DBLP:conf/cdc/AmesGT14}, the solution time is much faster than MPC. Additionally, it is natural for CBF-based controllers to consider constraint boundaries even when the system is far from them, thus, generating smooth control actions. Although CBFs have many advantages, constructing a valid CBF that recovers the true safe set is still challenging~\cite{DBLP:conf/cdc/ChoiLSTH21}\cite{ChenWC22}. However, finding an HCBF that recovers only a portion of the true safe set is relatively easy. Taking advantage of this insight, work has been done in learning-based CBF refinement starting from an HCBF~\cite{DaiKK22}. This approach makes constructing CBFs much easier, even for nonconvex safe sets. However, this approach requires a large dataset collected online, which can be both time-consuming and expensive.

Similar to learning CBFs, other learning-based methods also suffer from high sample complexity~\cite{WeiKK23}, e.g., deep reinforcement learning (DRL). Experience replay (ER) is a commonly used method in off-policy DRL algorithms~\cite{mnih2015human}\cite{LillicrapHPHETS15} to reduce the sample complexity by recycling data collected from previous episodes. When using ER, the data is stored in a replay buffer and sampled when updating the control policy. In the vanilla version of ER, the data is sampled uniformly. Prioritized experience replay (PER)~\cite{SchaulQAS15} is then proposed as a low sample complexity version of ER. For PER, the stored data is sampled based on a priority score. Data with higher priorities are more likely to be sampled, and vice versa. The use of PER greatly reduces the sample complexity of off-policy DRL algorithms~\cite{pan2022understanding}~\cite{FujimotoMP20}. In this work, we propose integrating a PER strategy into learning-based CBF refinement.

In this paper, we propose an algorithmic approach to enable lower sample complexity for learning a CBF starting from an HCBF. The main contribution of this paper is threefold: (1) we developed a method that integrates PER into the learning-based CBF refinement pipeline; (2) we provide a theoretical analysis of how PER affects the learning process; (3) we demonstrate the effectiveness of our approach using simulation studies. The remainder of this paper is structured as follows. In Section II, a background summary regarding CBFs and PER is provided. In Section III, the data-efficient CBF refinement problem is formulated. In Section IV, we describe the CBF learning algorithm along with the prioritized data sampling pipeline. In Section V, The effectiveness of our method is demonstrated in simulation on a unicycle and a two-link arm.. Finally, in Section VI, the paper is concluded.

%% file: preliminary.tex
\section{Preliminaries}
In this section, we present a brief introduction to CBF-based control and PER.

\subsection{Control Barrier Function}

Consider a set $\mathcal{C}$ that is the 0-superlevel set of a continuously differentiable function $h: \mathcal{C} \subset \mathcal{D} \rightarrow \mathbb{R}$, yielding
\begin{subequations}
\label{eq:cbf_conditions}
\begin{align}
    \mathcal{C} &= \{x\in\mathcal{D}\subset\mathbb{R}^n \mid h(x) \geq 0\}\\
    \partial\mathcal{C} &= \{x\in\mathcal{D}\subset\mathbb{R}^n \mid h(x) = 0\}\\
    \mathrm{Int}(\mathcal{C}) &= \{x\in\mathcal{D}\subset\mathbb{R}^n \mid h(x) > 0\}
\end{align}
\end{subequations}
where $\partial\mathcal{C}$ represents the boundary of $\mathcal{C}$ and $\mathrm{Int}(\mathcal{C})$ represents the interior of $\mathcal{C}$. Additionally, we assume that $\mathrm{Int}(\mathcal{C})$ is not an empty set, i.e., $\mathrm{Int}(\mathcal{C})\neq\emptyset$, and $\mathcal{C}$ does not contain any isolated points. We refer to this set $\mathcal{C}$ as the safe set. Then, consider a system in control affine form
\begin{equation}
    \dot{x} = f(x) + g(x)u
    \label{eq:control_affine_system}
\end{equation}
where the state is $x\in\mathbb{R}^n$, the control is $u\in\mathbb{R}^m$, the drift is $f: \mathbb{R}^n\rightarrow\mathbb{R}^n$, and the control influence matrix is $g: \mathbb{R}^n\rightarrow\mathbb{R}^{n\times m}$. Both $f$ and $g$ are also locally Lipschitz continuous. The closed-loop dynamics of the system in~\eqref{eq:control_affine_system} is
\begin{equation}
    \dot{x} = f_{\mathrm{cl}}(x) = f(x) + g(x)\pi(x)
    \label{eq:closed_loop_system}
\end{equation}
where the feedback controller is
\begin{equation}
    u = \pi(x)
\end{equation}
and $\pi: \mathbb{R}^n\rightarrow\mathbb{R}^{m}$ is locally Lipschitz continuous. For any initial state $x_0\in\mathbb{R}^n$, there exists an interval of existence
\begin{equation}
    I(x_0) = [t_0, t_{\mathrm{max}})
\end{equation}
such that $x(t)$ is the unique solution to~\eqref{eq:closed_loop_system} on $I(x_0)$. When $t_{\mathrm{max}} = \infty$, the system in~\eqref{eq:closed_loop_system} is forward complete. Using the concepts mentioned above, the notions of forward invariance, safety with respect to a set, and CBFs are defined as follows.

\begin{defn}[Forward Invariance \& Safety]
The system defined in~\eqref{eq:closed_loop_system} is forward invariant with respect to a set $\mathcal{C}\subset\mathbb{R}^n$ if for every $x_0\in\mathcal{C}$, we have $x(t)\in\mathcal{C}$ for all $t \in I(x_0)$. A system that is forward invariant with respect to $\mathcal{C}$ is said to be safe with respect to $\mathcal{C}$. A controller that makes a closed-loop system safe with respect to $\mathcal{C}$ is said to be safe with respect to $\mathcal{C}$.
\end{defn}

\begin{defn}[Control Barrier Function~\cite{DBLP:conf/eucc/AmesCENST19}]
Let $\mathcal{C}$ be the 0-superlevel set of a continuously differentiable function $h: \mathcal{D}\rightarrow\mathbb{R}$, with $\partial h/\partial x \neq 0$ for all $x\in\partial\mathcal{C}$. Then, $h$ is a control barrier function (CBF) on $\mathcal{C}$ if there exists an extended class $\mathcal{K}_\infty$ function $\alpha: \mathbb{R}\rightarrow\mathbb{R}$ such that for all $x\in\mathcal{D}$ the system defined in~\eqref{eq:control_affine_system} satisfies
\begin{equation}
    \sup_{u\in\mathcal{U}}\Big[\frac{\partial h(x)}{\partial x}\Big(f(x) + g(x)u\Big)\Big] \geq -\alpha(h(x))
    \label{eq:CBF_constraint}
\end{equation}
with $\mathcal{U}$ being the admissible set of controls.
\end{defn}

The CBF constraint in~\eqref{eq:CBF_constraint} is used in CBF-based quadratic programs (CBFQPs)~\cite{DBLP:conf/iccps/GurrietSRCFA18} that find the closest safe control action to a possibly unsafe control action
\begingroup
\allowdisplaybreaks
\begin{align}
    \min_{u\in\mathcal{U}}\ &\ \|u - \pi_\mathrm{perf}(x)\|^2\\
    \mathrm{subject\ to}\ &\ \Big[\frac{\partial h(x)}{\partial x}\Big(f(x) + g(x)u\Big)\Big] \geq -\alpha(h(x))\nonumber
\end{align}
\endgroup
with the possibly unsafe control action being generated by a performance controller $\pi_\mathrm{perf}: \mathbb{R}^n\rightarrow\mathbb{R}^m$. Assuming no control constraints (a common assumption in the CBF literature) and a valid CBF, the CBFQP will always be feasible, and the corresponding controller is Lipschitz continuous~\cite{DBLP:conf/cdc/MorrisPA13}.

\subsection{Prioritized Experience Replay}
Unlike vanilla ER, where the data is sampled uniformly, for PER, the data are sampled based on a priority score $p$. In off-policy DRL algorithms, the priority score is commonly chosen as the temporal difference (TD) error
\begin{equation}
    \mathrm{TD\ Error} = Q - (r + Q^\prime)
\end{equation}
with $Q\in\mathbb{R}$ being the action-value function at this state, $r\in\mathbb{R}$ being the reward, and $Q^\prime\in\mathbb{R}$ the action-value function at the next state. Then, the sampling probability for each data point can be computed as
\begin{equation}
    P(i) = \frac{p_i^{\alpha_p}}{\sum_{i=0}^{N}p_i^{\alpha_p}}
\end{equation}
where $P(i)$ represents the probability of sampling the $i$-th data point, $p_i$ is the priority score for the $i$-th data point, $N$ is the number of data points within the replay buffer, and $\alpha_p$ regulates how PER behaves, with $\alpha_p = 0$ corresponding to uniform sampling and $\alpha_p = 1$ corresponding to full prioritized sampling. The implementation details of PER will be presented in Section~\ref{sec:prioritized_sampling}.

%% file: formulation.tex
\section{Problem Formulation}
In this section, we present the problem formulation for data-efficient CBF refinement. Consider the constraints
\begin{equation}
    c_i(x) \geq 0
    \label{eq:state_constraints}
\end{equation}
with $c_i:\mathbb{R}^n \rightarrow \mathbb{R}$ and $i = 1, \cdots, n_c$. Define $\mathcal{S}_i$ as
\begin{equation}
    \mathcal{S}_i = \{x \mid c_i(x) > 0\}
\end{equation}
and the intersections of all $\mathcal{S}_i$'s as $\mathcal{S}$
\begin{equation}
    \mathcal{S} = \bigcap_{i=1}^{n_c}{\mathcal{S}_i}.
\end{equation}
Then, the safe set $\mathcal{C}$ is defined as the largest forward invariant set in $\mathcal{S}$, which leads to the relationship
\begin{equation}
    \mathcal{C} \subseteq \mathcal{S}.
\end{equation}
Given the true safe set of our problem $\mathcal{C}$, we assume access~\cite{DaiKK22} to an HCBF $\widehat{h}: \mathbb{R}^n\rightarrow\mathbb{R}$ that is continuously differentiable, locally Lipschitz continuous and satisfies
\begin{equation}
    \widehat{\mathcal{C}} = \{x \mid \widehat{h}(x) \geq 0\} \subseteq \mathcal{C}.
\end{equation}
Without loss of generality, with $h$ defined in~\eqref{eq:cbf_conditions}, we can assume the relationship
\begin{equation}
\label{eq:cbf_relationship}
    h(x) = \widehat{h}(x) + \Delta{h}(x)
\end{equation}
with $\Delta h: \mathbb{R}^n \rightarrow \mathbb{R}$ being a continuously differentiable and locally Lipschitz continuous function. Additionally, we assume that $\widehat{h}$ and $\Delta{h}$ have the same relative degree as $h$. This is a mild assumption given that the relative degree of systems can be inferred using first principles. In this paper, we aim to find an algorithmic approach that learns $\Delta{h}$ in a {\em data-efficient} manner.

%% file: method.tex
\section{Method}
In this section, we present our proposed approach for the data-efficient refinement of CBFs. First, we describe the process of learning-based CBF refinement. Then, we show how to improve the CBF learning procedure using PER. Finally, we discuss how to evaluate the effectiveness of our proposed approach.

\subsection{Learning-Based CBF Refinement}
Since the only unknown in~\eqref{eq:cbf_relationship} is $\Delta{h}$, we can use a neural network  parameterized by $\theta$ to estimate $\Delta{h}(x)$, which is denoted as $\Delta\widehat{h}(x\mid\theta)$. Then, the learned CBF is defined as
\begin{equation}
    \widetilde{h}(x\mid\theta) = \widehat{h}(x) + \Delta\widehat{h}(x\mid\theta).
\end{equation}
We use a deep differential network (DDN)~\cite{LutterRP19} to represent $\Delta\widehat{h}(x\mid\theta)$, which given input $x$, outputs both $\Delta\widehat{h}(x\mid\theta)$ and the analytical $\partial\Delta\widehat{h}(x\mid\theta)/\partial x$. If smooth activation functions are used, DDNs will be continuously differentiable. Another challenge when learning the CBF is having no ground truth CBF values. However, it is easier to recognize whether a state is safe using the constraints in~\eqref{eq:state_constraints}. We then write loss functions to enforce the conditions in~\eqref{eq:cbf_conditions}, yielding,
\begin{subequations}
\label{eq:old_loss_functions}
\begin{align}
    \mathcal{L}_+(\theta) &= \frac{1}{B}\sum_{x_i\in\mathcal{X}_+}^{}\max\Big(0, -\widetilde{h}(x_i\mid\theta)\Big)\\
    \mathcal{L}_-(\theta) &= \frac{1}{B}\sum_{x_i\in\mathcal{X}_-}^{}\max\Big(0, \widetilde{h}(x_i\mid\theta)\Big)
\end{align}
\end{subequations}
with $B$ representing the batch size, $\mathcal{L}_+$ representing the loss for safe states, $\mathcal{L}_-$ representing the loss for unsafe states, $\mathcal{X}_{+}$ being the dataset containing safe interactions, and $\mathcal{X}_{-}$ being the dataset containing unsafe interactions. It is much easier to determine whether a state is safe using~\eqref{eq:state_constraints}. One major issue with the loss functions in~\eqref{eq:old_loss_functions} is that a sample $x$ that satisfies
\begin{equation}
    \widetilde{h}(x\mid\theta) \equiv 0
\end{equation}
gives zero loss. To better guide the learning, a distance function $d(x)$ is introduced~\cite{DaiKK22} that increases while going inside the interior of the safe set and decreases while going inside the interior of the unsafe set. Using the distance function, we have the new loss function as
\begin{subequations}
\label{eq:safe_unsafe_loss}
\begin{align}
    \mathcal{L}_+(\theta) &= \frac{1}{B}\sum_{x_i\in\mathcal{X}_+}^{}\max\Big(0, -\widetilde{h}(x_i\mid\theta) + d(x_i)\Big)\\
    \mathcal{L}_-(\theta) &= \frac{1}{B}\sum_{x_i\in\mathcal{X}_-}^{}\max\Big(0, \widetilde{h}(x_i\mid\theta) - d(x_i)\Big).
\end{align}
\end{subequations}
Additionally, to ensure that control actions exists at each point within the safe set that satisfies~\eqref{eq:CBF_constraint}, we use the loss function
\begingroup
\allowdisplaybreaks
\begin{align}
\label{eq:cbf_constraint_loss}
    \mathcal{L}_{\nabla h}(\theta) =&\ \frac{1}{B}\sum_{x_i\in\mathcal{X}_\nabla}^{}\max\Big(0, \frac{\partial\widetilde{h}(x_i\mid\theta)}{\partial x}\dot{x}_i\nonumber\\
    &- \alpha(\widetilde{h}(x_i\mid\theta))\Big)
\end{align}
\endgroup
with $\mathcal{X}_\nabla$ being the replay buffer for the CBF constraint loss. Then, we have the final loss function as
\begin{equation}
\label{eq:cbf_loss}
    \mathcal{L}(\theta) = \mathcal{L}_+(\theta) + \lambda\mathcal{L}_-(\theta) + \mathcal{L}_{\nabla h}(\theta)
\end{equation}
where $\lambda\in\mathbb{R}_+$ weights the importance of the unsafe loss. Since it is much more dangerous to recognize an unsafe state as a safe one than the other way around, usually $\lambda > 1$.

The training process is separated into multiple epochs, where each epoch consists of one episode and multiple neural network updates. During each episode, a CBFQP controller using the learned CBF is used to collect state-action tuples, which are then stored in the replay buffers. After each episode is finished, data are sampled from the replay buffers in batches to compute the loss function in~\eqref{eq:cbf_loss} and its gradient. Then, using a stochastic gradient descent (SGD) algorithm, e.g., ADAM~\cite{KingmaB14}, the weight vector of the neural network is updated. For each update, new batches of data will be sampled. Note that, during training, the system would encounter unsafe states. We can utilize simulation environments (e.g., PyBullet or MuJoCo) or create safe versions of the real scenario (e.g., replace obstacles with virtual projections) to avoid endangering the real system.

\subsection{Prioritized Sampling}
\label{sec:prioritized_sampling}
In this section, we will discuss how to incorporate a PER buffer into the CBF learning pipeline. To utilize PER, we first need to construct a priority score. For the unsafe buffer, we choose
\begin{equation}
\label{eq:unsafe_priority}
    p_i^- = \max\Big(0, \widetilde{h}(x_i\mid\theta) - d(x_i)\Big).
\end{equation}
For the safe buffer, we choose
\begin{equation}
\label{eq:safe_priority}
    p_i^+ = \max\Big(0, -\widetilde{h}(x_i\mid\theta) + d(x_i)\Big).
\end{equation}
For the CBF constraint buffer, we choose
\begin{equation}
    p_i^\nabla = \max\Big(0, \frac{\partial\widetilde{h}(x_i\mid\theta)}{\partial x}\dot{x}_i - \alpha(\widetilde{h}(x_i\mid\theta))\Big).
\end{equation}
To make the prioritized buffer computationally more efficient, we can use a sum tree data structure~\cite{SchaulQAS15} as shown in Fig.~\ref{fig:PER}, which is popular among PER~\cite{SchaulQAS15} implementations. In the sum tree, each leaf node represents a data point with a value corresponding to the data point priority. Each parent node has two children. The parent node value is the sum of the children node's values. The value of the root node is the sum of all leaf node values, which corresponds to $\sum_{i=1}^{N}{p_i}$. When adding a data point, it occupies the next available leaf node, and the values of all of its parent nodes are updated. This update has a time complexity of $O(\log{n})$. When sampling a data point, a random number $\xi$ is sampled from $[0, \sum_N p_i]$. Then, starting from the root node, $\xi$ is first compared with the value of the left child node; if smaller, the left child node is considered; if larger, $\xi \leftarrow \xi - p_l$, with $p_l$ being the value of the left children node, and the right child node is considered. This operation is recursively applied until one of the leaf nodes is reached. This process is equivalent to drawing a line from zero to $\sum_N p_i$ and letting each leaf node occupy a space with length $p_i$. Then the probability of $\xi$ landing inside the portion belonging to node $i$ has the probability of $p_i / \sum_N p_i$, which is illustrated in Fig.~\ref{fig:PER}. 

\begin{figure}[t!]
    \centering
    \includegraphics[width=0.49\textwidth]{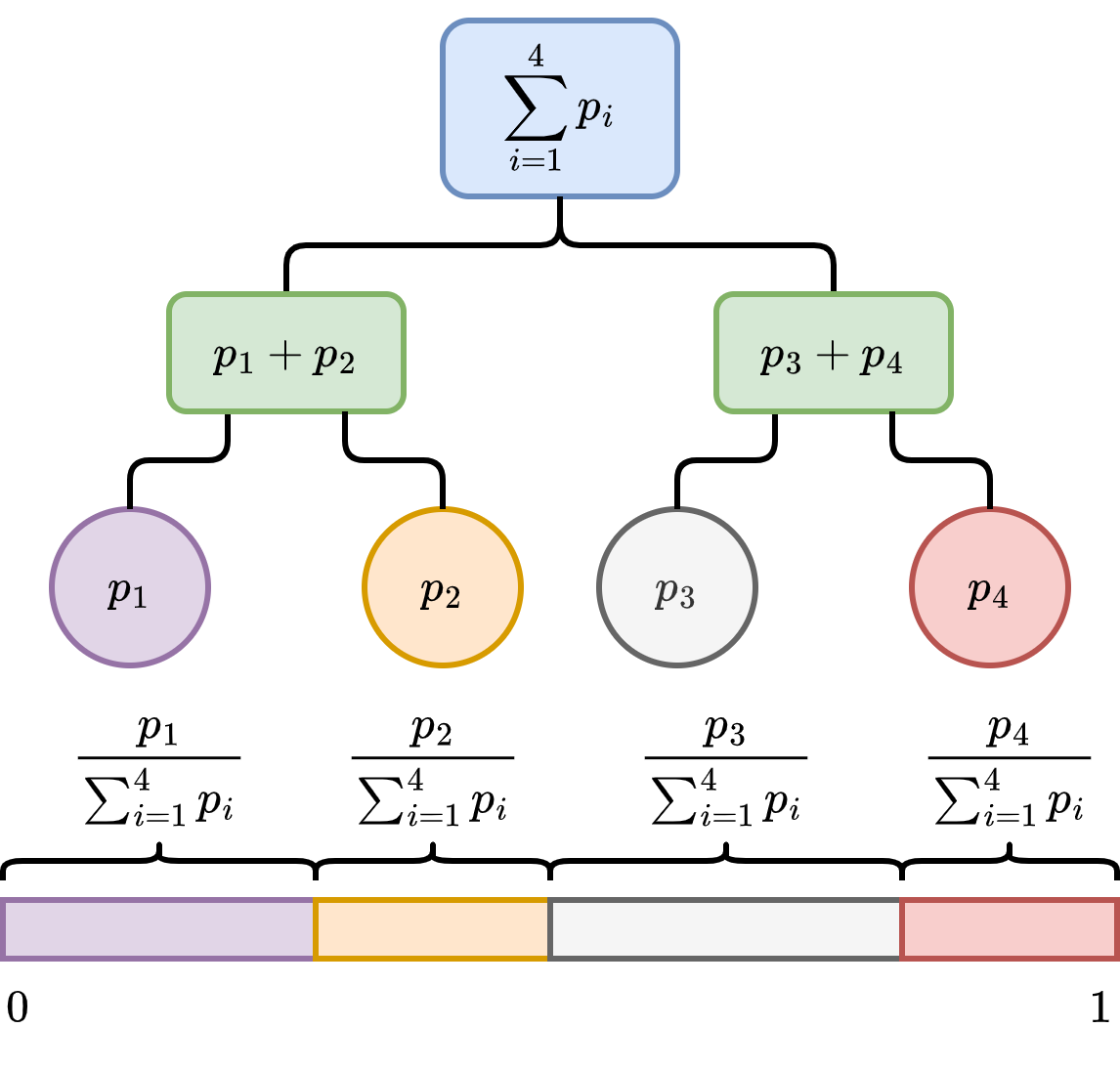}
    \caption{This figure illustrates the structure of a sum tree with four leaf nodes and its corresponding probability distribution.}
    \label{fig:PER}
\end{figure}

Following the formulation in~\cite{SchaulQAS15}, a fake priority is assigned when data is first added to the buffer; the priority is then updated when the loss function is computed. The fake priority is chosen to be $p_\mathrm{fake} = p_\mathrm{max} + \epsilon$ with $p_\mathrm{max}$ being the largest priority currently present in the PER buffer and $\epsilon\in\mathbb{R}_+$ being a small value. This choice of $p_\mathrm{fake}$ ensures new data has the largest priority. Additionally, the relationship
\begin{equation}
\label{eq:update_relationship}
    BN_U \geq T
\end{equation}
must hold, with $N_U$ being the number of updates per epoch and $T$ being the number of time steps within each episode. These conditions ensure each new data point will be sampled at least once and an actual priority value is assigned. Next, we discuss how the prioritized sampling scheme can expedite learning in the CBF refinement setting.
\begin{rem}
When using~\eqref{eq:safe_unsafe_loss} and~\eqref{eq:cbf_constraint_loss}, the replay buffer can be separated into two parts, one that has zero losses and the other that has nonzero losses. For the zero loss portion, the sampling priority is the lowest, and the corresponding gradient value is zero.
\end{rem}
\begin{rem}
For the nonzero loss portion, \eqref{eq:safe_unsafe_loss} and~\eqref{eq:cbf_constraint_loss} have the same form as
\begin{equation}
    \mathcal{L}_{L1} = \frac{1}{\bar{N}}\sum_{i=1}^{\bar{N}}|F_\theta(x_i) - y_i| 
\end{equation}
with $\bar{N}\in\mathbb{Z}_+$ being the size of the replay buffer, $F_\theta: \mathbb{R}^n\rightarrow\mathbb{R}^{n_o}$ the neural network, $y_i\in\mathbb{R}^{n_o}$ the training labels, and $n_o$ the dimension of the output ($n_o = 1$ in our case). 
\end{rem}

\begingroup
\allowdisplaybreaks
\begin{thm}
\label{thm:key}
For a constant $c$ that depends on $\theta$, we have
\begin{align}
    &\ c\cdot\mathbb{E}_{(x, y)\sim\mathrm{Prioritized}}\Big[\nabla_\theta|F_\theta(x) - y|\Big]\nonumber\\
    &= \mathbb{E}_{(x, y)\sim\mathrm{Uniform}}\Big[\frac{1}{2}\nabla_\theta(F_\theta(x) - y)^2\Big]
\end{align}
\end{thm}
\endgroup

\begin{proof}
\begingroup
\allowdisplaybreaks
\begin{align}
    &\ \mathbb{E}_{(x, y) \sim \mathrm{Prioritized}}\Bigg[\frac{\partial|F_\theta(x) - y|}{\partial\theta}\Bigg]\nonumber\\
    =\ &\ \frac{1}{\sum_{i = 1}^{\bar{N}}|F_\theta(x_i) - y_i|}\sum_{i=1}^{\bar{N}}|F_\theta(x_i) - y_i|\frac{\partial|F_\theta(x_i) - y_i|}{\partial\theta}\nonumber\\
    =\ &\ \frac{1}{\sum_{i = 1}^{\bar{N}}|F_\theta(x_i) - y_i|}\sum_{i=1}^{\bar{N}}|F_\theta(x_i) - y_i|\nonumber\\
    &\ \cdot\frac{\partial\Big((F_\theta(x_i) - y_i)^2\Big)^{\frac{1}{2}}}{\partial(F_\theta(x_i) - y_i)^2}\frac{\partial(F_\theta(x_i) - y_i)^2}{\partial\theta}\nonumber\\
    =\ &\ \frac{1}{2\sum_{i = 1}^{\bar{N}}|F_\theta(x_i) - y_i|}\sum_{i=1}^{\bar{N}}\frac{\partial(F_\theta(x_i) - y_i)^2}{\partial\theta}.
\end{align}
Then, setting
\begin{equation}
    c = \frac{1}{\bar{N}}\sum_{i = 1}^{\bar{N}}|F_\theta(x_i) - y_i|
\end{equation}
completes the proof.
\end{proof}
\endgroup

Theorem~\ref{thm:key} states that by using prioritized sampling, the loss function is equivalent to an L2 loss for the non-zero loss portion, while the loss function is an L1 loss when using uniform sampling, as in~\cite{RobeyHLZDTM20}.

\begin{thm}
\label{thm:L1L2}
Define the absolute prediction error (APE) for the $i$-th data point in the replay buffer under L1 loss with learning rate $\eta\in\mathbb{R}_+$ at time $t$ as
\begin{equation}
    \delta_t(i) = |F_{\theta_t}(x_i) - y_i|,
\end{equation}
and under L2 loss with learning rate $\eta$ at time $t$ as
\begin{equation}
    \widetilde\delta_t(i) = |\widetilde F_{\theta_t}(x_i) - y_i|.
\end{equation}
Let $F_{\theta_0}(x) = \widetilde F_{\theta_0}(x)$. If $\delta_0(i) = \widetilde\delta_0(i) \geq 1$, $\forall i = 1, \cdots, \bar{N}$, then there exists $\epsilon(i) \in (0, 1]$, such that for an initial APE $\delta_0$ that satisfies $\epsilon(i) \leq \delta_0$, the time it takes to be reduced to $\epsilon(i)$ under L1 loss, defined as
\begin{equation}
    t_\epsilon = \min_t\{t \geq 0: \delta_t(i) \leq \epsilon(i)\},
\end{equation}
is longer than the time it takes under L2 loss, defined as
\begin{equation}
    \widetilde{t}_\epsilon = \min_t\{t \geq 0: \widetilde{\delta}_t(i) \leq \epsilon(i)\},
\end{equation}
i.e., $\widetilde{t}_\epsilon \leq t_\epsilon$.
\end{thm}

\begin{proof}
Given
\begin{equation}
    \frac{\partial\delta_t(i)}{\partial F_{\theta_t}(x_i)} = \frac{\partial\widetilde\delta_t(i)}{\partial F_{\theta_t}(x_i)} = \frac{\partial\mathcal{L}_{L1}(F_{\theta_t}, y)}{\partial F_{\theta_t}(x_i)} = \mathrm{sgn}\{F_{\theta_t}(x_i) - y_i\}
\end{equation}
and
\begin{equation}
    \frac{\partial\mathcal{L}_{L2}(F_{\theta_t}, y)}{\partial F_{\theta_t}(x_i)} = F_{\theta_t}(x_i) - y_i,
\end{equation}
we have
\begingroup
\allowdisplaybreaks
\begin{align}
\label{eq:ddelta_t_dt}
    \frac{d\delta_t(i)}{dt} &= {\frac{\partial\delta_t(i)}{\partial F_{\theta_t}(x_i)}\frac{dF_{\theta_t}(x_i)}{dt}}\nonumber\\
    &= {\frac{\partial \delta_t(i)}{\partial F_{\theta_t}(x_i)}\frac{\partial F_{\theta_t}(x_i)}{\partial\theta_t}(-\eta)\Big(\frac{\partial\mathcal{L}_{L1}(F_{\theta_t}, y)}{\partial F_{\theta_t}(x_i)}\frac{\partial F_{\theta_t}(x_i)}{\partial\theta_t}\Big)^T}\nonumber\\
    &= -\eta\frac{\partial F_{\theta_t}(x_i)}{\partial\theta_t}\Big(\frac{\partial F_{\theta_t}(x_i)}{\partial\theta_t}\Big)^T
\end{align}
\endgroup
and
\begin{align}
    \frac{d\widetilde\delta_t(i)}{dt} &= \frac{\partial\widetilde\delta_t(i)}{\partial F_{\theta_t}(x_i)}\frac{dF_{\theta_t}(x_i)}{dt}\nonumber\\
     &= {\frac{\partial\widetilde\delta_t(i)}{\partial F_{\theta_t}(x_i)}\frac{\partial F_{\theta_t}(x_i)}{\partial\theta_t}(-\eta)\Big(\frac{\partial\mathcal{L}_{L2}(F_{\theta_t}, y)}{\partial F_{\theta_t}(x_i)}\frac{\partial F_{\theta_t}(x_i)}{\partial\theta_t}\Big)^T}\nonumber\\
    &= -\eta\widetilde\delta_t(i)\frac{\partial F_{\theta_t}(x_i)}{\partial\theta_t}\Big(\frac{\partial F_{\theta_t}(x_i)}{\partial\theta_t}\Big)^T
\end{align}
with $\partial F_{\theta_t}(x_i)/\partial\theta_t$ being a row vector. This shows that when $\widetilde\delta_t(i) > 1$, the APE will decrease faster under L2 loss than under L1 loss. Then, based on the comparison lemma~\cite{khalil2015nonlinear}, if $\delta_0(i)$ is greater than one, the APE under L2 loss will reach APE equal to one faster than under L1 loss. The two APEs will intersect at some value lower than one. Thus, for some value $\epsilon(i) \in (0, 1]$, we have the relationship $\widetilde t_\epsilon \leq t_\epsilon$.
\end{proof}

\begin{rem}
For the CBF losses in~\eqref{eq:safe_unsafe_loss} and~\eqref{eq:cbf_constraint_loss}, when $\delta_0(i) < 1$, L1 loss converges faster to zero than L2 loss.
\end{rem}

\begin{rem}
For the CBF losses in~\eqref{eq:safe_unsafe_loss} and~\eqref{eq:cbf_constraint_loss}, the goal is not to converge to zero but go past zero. For example, in the case of $\mathcal{L}_{-}$, the goal is not to make $\widetilde{h}(x_i\mid\theta) - d(x_i) = 0$ but rather $\widetilde{h}(x_i\mid\theta) - d(x_i) \leq 0$.
\end{rem}

Theorem~\ref{thm:L1L2} shows when APE is larger than one, instantaneously, L2 loss converges faster, and when the APE is less than one, L1 loss converges faster. In the learning-based CBF refinement setting, using the proposed prioritized sampling and priority assignment scheme, the effect is a combination of both L1 and L2 loss. To understand this effect, we show that the data sampling process within one epoch can be seen as having two phases. For the first phase, the majority of the data sampled are newly collected data points, and according to our empirical studies, the majority have an APE of less than one. In the second phase, after most of the new data points are sampled, if~\eqref{eq:update_relationship} is satisfied, data points with the largest APEs will more likely get sampled. If $B$ and $N_U$ are well chosen, the amount of data sampled in the second phase with an APE lower than one will be small, and the majority of the data sampled with have an APE larger than one. Therefore, this data collection and sampling scheme can be approximately seen as applying L1 loss with uniform sampling to low APE data points while applying L1 loss with prioritized sampling (i.e., effectively L2 loss) to high APE data points. This makes our proposed approach faster than using either L1 or L2 loss with uniform sampling.

\begin{figure*}[t!]
    \centering
    \includegraphics[width=\textwidth]{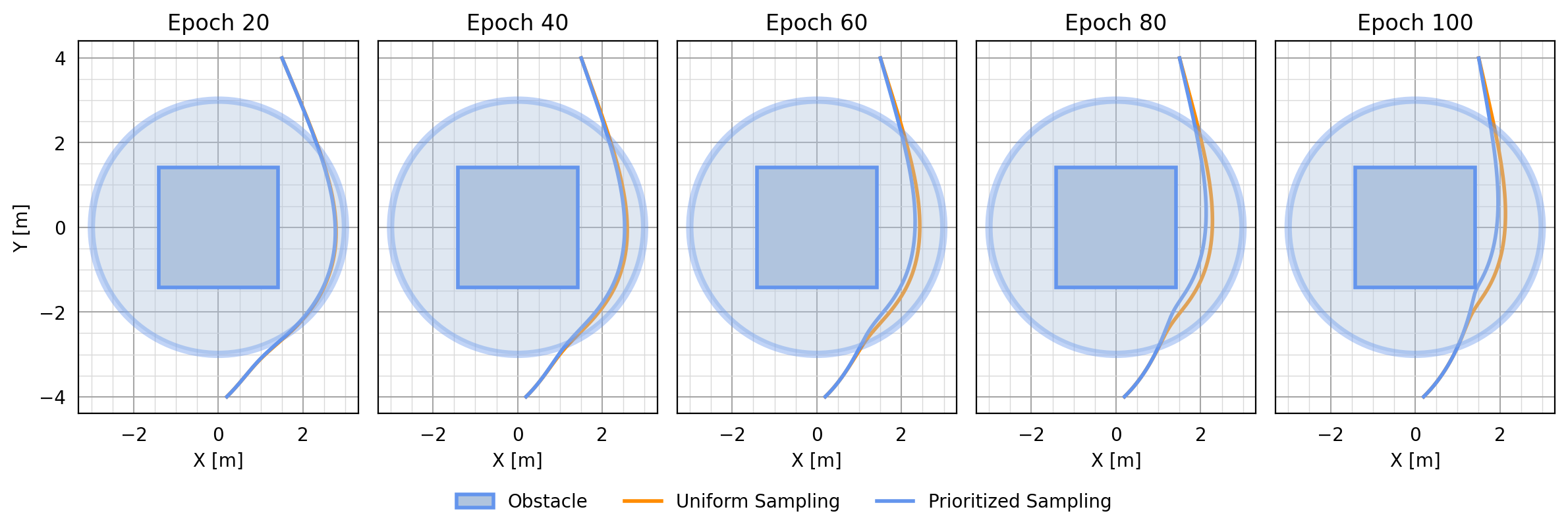}
    \caption{This figure compares the state trajectories generated by the learned controller using uniform sampling and prioritized sampling for the unicycle system. The darker blue square represents the square obstacle. The light blue circle represents the unsafe region corresponding to the HCBF.}
    \label{fig:unicycle_training}
\end{figure*}

\begin{figure}[t!]
    \centering
    \includegraphics[width=0.49\textwidth]{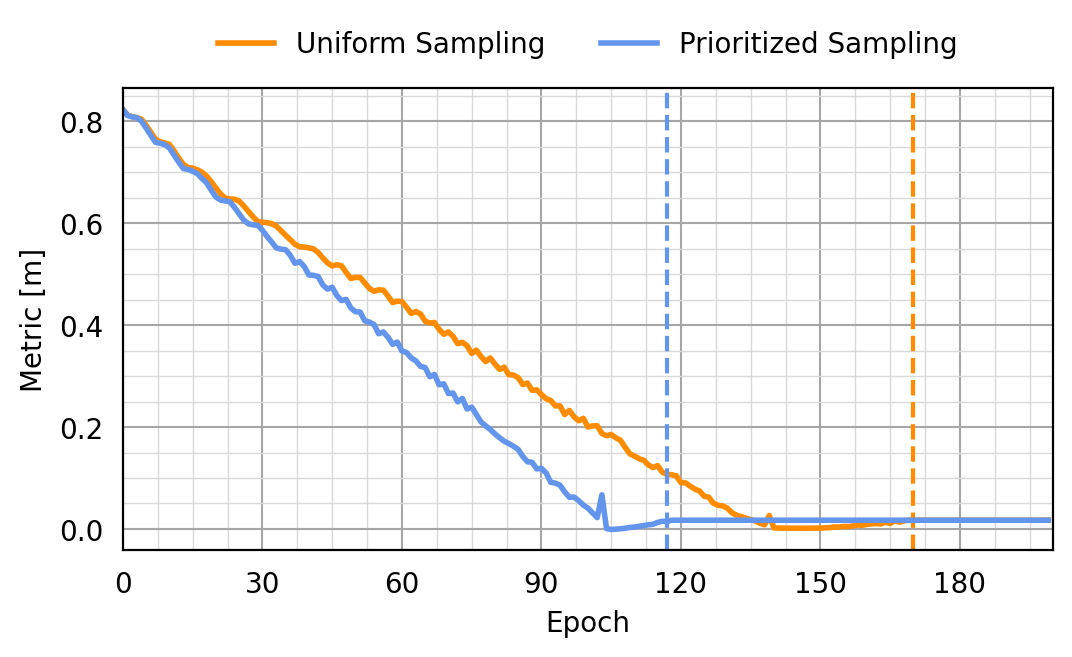}
    \caption{This figure compares the evaluation metric for using uniform and prioritized sampling over the training process on a unicycle system. The vertical dashed lines denote the epoch of convergence.}
    \label{fig:unicycle_training_metric}
\end{figure}

\subsection{Evaluation Metric}
To quantify the effectiveness of our proposed approach, we would need to find a metric that represents the size of the learned safe set. We use an approach inspired by Hamilton-Jacobi (HJ) reachability-based methods~\cite{BansalCHT17}. First, we define a distance between the current state and the safe set boundary 
\begin{equation}
   \ell(x) = \min\{c_1(x), \cdots, c_r(x)\}.
\end{equation}
Then, the evaluation metric is defined as
\begin{equation}
    m_e = \min_{t\in[0, T]}\Big\{\ell(x(t)) \mid x(0)\in\mathcal{X}_0, \dot{x}(t) = f_{cl}(x\mid\theta)\Big\} 
\end{equation}
with $\mathcal{X}_0$ being the set of initial states and $f_{cl}: \mathbb{R}^n\rightarrow\mathbb{R}^n$ the closed-loop system with control actions generated by the learned CBF controller.

%% file: simulation.tex
\section{Simulation}
In this section, we show the effectiveness of our approach using simulation studies. 

\subsection{Unicycle}
The task is to reach a target state while avoiding collision with a square obstacle. The unicycle system has the dynamics
\begin{equation}
\label{eq:unicycle_dynamics}
    \begin{bmatrix}
        \dot{x}\\
        \dot{y}\\
        \dot{\phi}
    \end{bmatrix} = \begin{bmatrix}
        \cos\phi & 0\\
        \sin\phi & 0\\
        0 & 1
    \end{bmatrix}\begin{bmatrix}
        v\\
        \omega
    \end{bmatrix}
\end{equation}
with $x$ denoting the position of the unicycle along the $x$ axis, $y$ denoting the position along the $y$ axis, and $\phi$ denoting the heading of the unicycle. The control inputs are the linear velocity $v$ and the angular velocity $\omega$. Given that it is difficult to construct an HCBF for a square obstacle, we write an HCBF for a circular obstacle~\cite{DBLP:journals/corr/abs-2110-05415} that contains the square obstacle
\begin{equation}
    \widehat{h}(\x) = x^2 + y^2 + 2xl\cos\phi + 2yl\sin\phi + l^2 - r^2
\end{equation}
where $r$ is the radius of the constructed circular obstacle and $l$ is a predefined lookahead distance. This HCBF corresponds to the safe set under the constraint
\begin{equation}
    x^2 + y^2 \geq r^2.
\end{equation}
During training, we use a proportional controller as the performance controller
\begin{equation}
    \pi(x, y, \phi) = \begin{bmatrix}
        K_ve\\
        K_\omega(\beta - \phi)
    \end{bmatrix}
\end{equation}
with $K_v$ and $K_\omega$ being the control gains, and
\begingroup
\allowdisplaybreaks
\begin{subequations}
\begin{align}
    e &= \sqrt{(x - x_{\mathrm{des}})^2 + (y - y_{\mathrm{des}})^2}\\
    \beta &= \mathrm{atan}2(y_{\mathrm{des}} - y, x_{\mathrm{des}} - x).
\end{align}
\end{subequations}
\endgroup
In our experiments, we set $K_v = 0.75$ and $K_\omega = 3.0$. The distance function is chosen to be
\begin{equation}
    d(x) = \max(|x|, |y|) - s/2
\end{equation}
with $s$ being the side length of the square obstacle.  For this example, $\ell(x)$ is the same as the distance function. We trained two learned CBF controllers, one using vanilla ER and the other one using PER. To ensure a fair comparison, we fixed all random seeds. During training, we set $\lambda = 100$, and use a learning rate of $10^{-4}$. We trained the learned controllers for 200 epochs, where at each epoch, we collected one episode of data (each episode has 1000 time steps) and updated the learned controllers 10 times using a batch size of 128. The state trajectory generated by the two learned controllers at different epochs is shown in Figure~\ref{fig:unicycle_training}, and the evaluation metric computed at each epoch is shown in Fig.~\ref{fig:unicycle_training_metric}. It can be easily seen from both figures that the controller using prioritized sampling is able to enlarge the estimated safe set faster. Additionally, we can see from Fig.~\ref{fig:unicycle_training_metric}, that after convergence, the evaluation metric for both of the controllers is the same. Thus, we can conclude that using prioritized sampling expedited the training while not making the performance of the learned controller deteriorate.

\subsection{Two-Link Arm}

\begin{figure*}[t!]
    \centering
    \includegraphics[width=\textwidth]{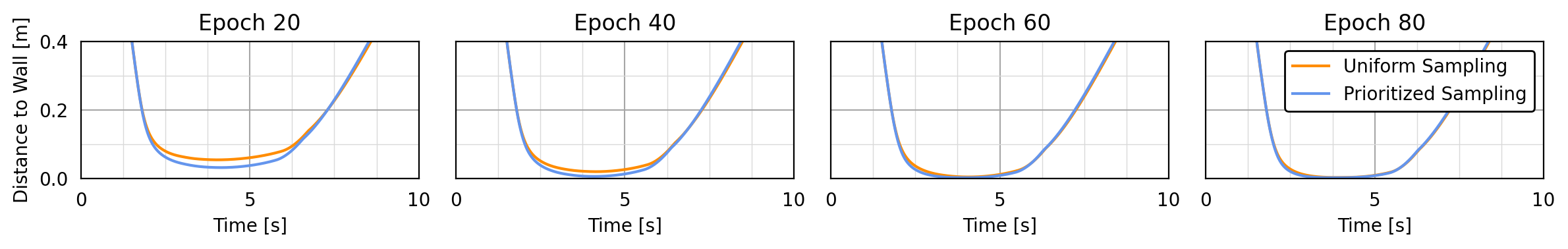}
    \caption{This figure compares the end-effector trajectories generated by the learned controller using uniform sampling and prioritized sampling.}
    \label{fig:two_link_arm_training}
\end{figure*}

\begin{figure}[t!]
    \centering
    \includegraphics[width=0.49\textwidth]{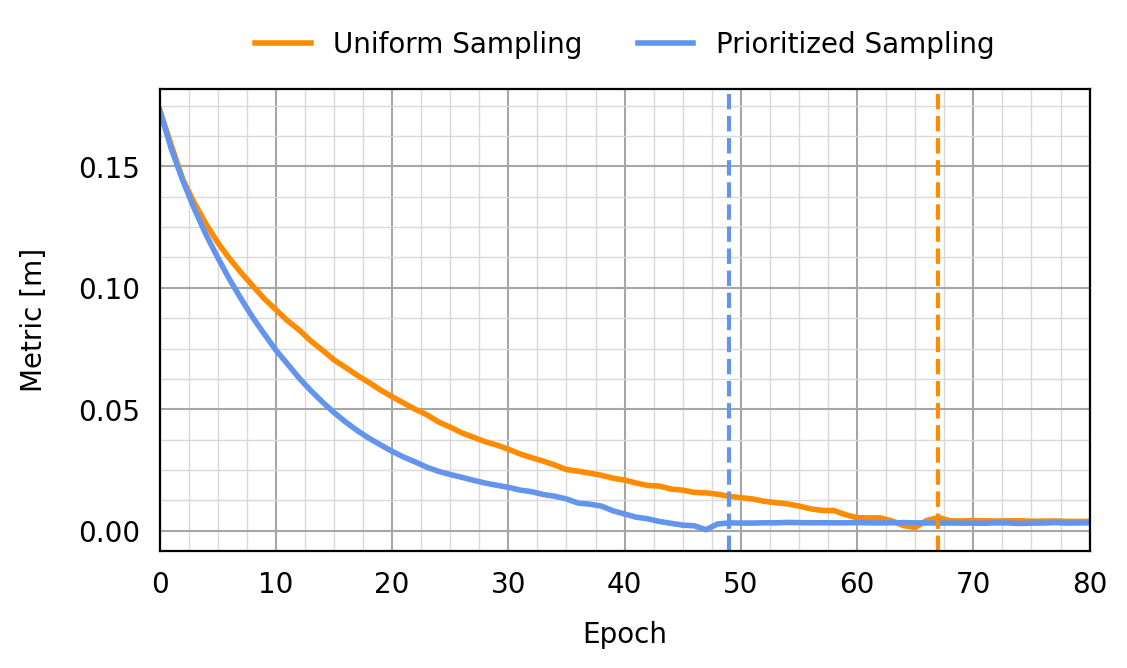}
    \caption{Comparison of the evaluation metric between using uniform and prioritized sampling over the training process on a two-link arm system. The vertical dashed lines denote the epoch of convergence.}
    \label{fig:two_link_arm_prioritized_sampling}
\end{figure}

\begin{figure}[t!]
    \centering
    \includegraphics[width=0.49\textwidth]{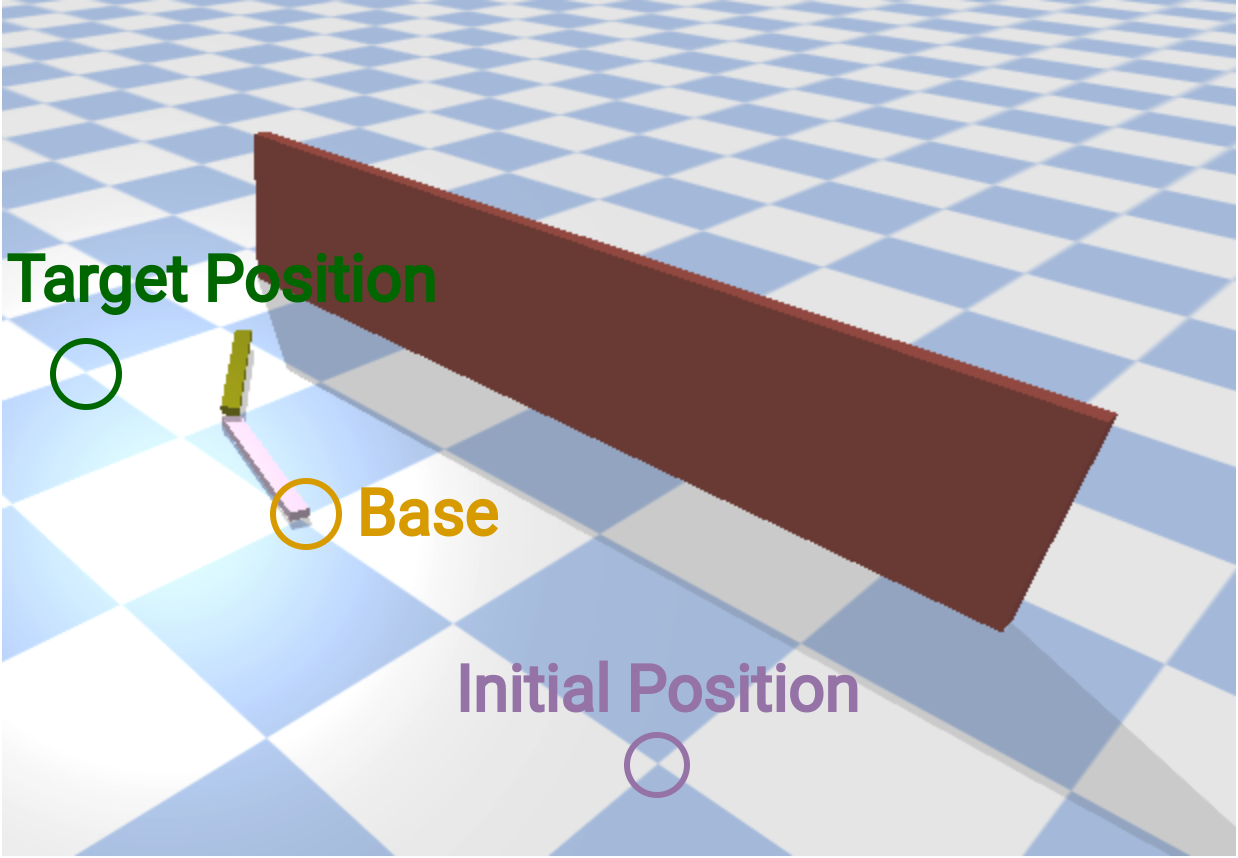}
    \caption{Screenshot of the two-link arm simulation environment. The task is to move from the initial position to the target position while not colliding with the wall. The first link is light purple and the second link is dark green. Each tile on the ground is $1m\times1m$.}
    \label{fig:two_link_arm_sim}
\end{figure}

The two-link arm system has the dynamics
\begin{equation}
    \begin{bmatrix}
        \dot{q}\\
        \ddot{q}
    \end{bmatrix} = \begin{bmatrix}
        \dot{q}\\
        -M^{-1}(q)[C(q, \dot{q}) + G(q)]
    \end{bmatrix} + \begin{bmatrix}
        0\\
        M^{-1}(q)
    \end{bmatrix}\tau
\end{equation}
with $q\in\mathbb{R}^2$ representing the joint angles, $\tau\in\mathbb{R}^2$ representing the joint torques, $M(q)\in\mathbb{R}^{2\times2}$ representing the inertia matrix (which is positive definite), $C(q, \dot{q})\in\mathbb{R}^{2\times1}$ representing the Coriolis and centrifugal terms, and $G(q)\in\mathbb{R}^{2\times1}$ representing the gravitational terms. The task is to move the end-effector from the initial position to the target position while avoiding collision with a wall, which is illustrated in Fig.~\ref{fig:two_link_arm_sim}. The link length is $1$m for both of the links. The wall is positioned at $y = 1.75$m. Both the two-link arm and the wall can be seen as convex polytopes. Thus, the task is a polytopic obstacle avoidance task. To deal with this task, work has been done using discrete-time-CBF-based MPC~\cite{AkshayZS22} and finding a conservative kinematic controller~\cite{SingletaryGMSA22}. In this example, we learn a CBF-based dynamic controller to solve this task.

Given the difficulty in constructing an HCBF for polytopic constraints, we construct an HCBF that constraints the position of the end-effector to be less than $\beta$ along $\bar{n}\in\mathbb{R}^2$
\begin{equation}
\label{eq:two-link-arm-hcbf}
    \widehat{h}(x) = -\bar{n}^TJ(q)\dot{q} + \gamma(\beta - \bar{n}^T\mathrm{FK}(q))
\end{equation}
with $J(q)\in\mathbb{R}_{3\times 2}$ being the Jacobian matrix for the end-effector, $\mathrm{FK}: \mathbb{R}^2\rightarrow\mathbb{R}^2$ the forward kinematics function for the end-effector, $\bar{n}\in\mathbb{R}^2$ the normal vector of the edge of the polytopic obstacle, $d_\mathrm{max}\in\mathbb{R}$ the distance between the base of the two-link arm and the obstacle edge along $\bar{n}$, and $\gamma\in\mathbb{R}_+$. In this example $n = 2$. The HCBF in~\eqref{eq:two-link-arm-hcbf} corresponds to the constraint
\begin{equation}
    p_{ee}^T\bar{n} \leq \beta
\end{equation}
with $p_{ee}\in\mathbb{R}^2$ representing the position of the end-effector. For this example, we choose $\beta = 1.5$ and $\bar{n} = [0, 1]^T$. During training, we set
\begin{equation}
    d(x) = \mathrm{Distance}(x) - \mathrm{Penetration}(x)
\end{equation}
where $\mathrm{Distance}(x)$ is computed using the Gilbert Johnson Keerthi (GJK) algorithm~\cite{gilbert1988fast} and $\mathrm{Penetration}(x)$ is computed using the expanding polytope algorithm (EPA)~\cite{van2001proximity}. We choose $\ell(x)$ to be the same as the distance function. To ensure safety during training, we artificially enlarge the unsafe set and treat all interactions within the new unsafe set as unsafe. During training, we set $\lambda = 100$, and we use a learning rate of $10^{-4}$. We trained the learned controllers for 80 epochs, where at each epoch, we collected one episode of data (each episode has 3000 time steps) and updated the learned controllers 30 times using a batch size of 128. The value of $d(x)$ along the generated trajectory at different epochs is shown in Fig.~\ref{fig:two_link_arm_training}, and the evaluation metric over the learning process is shown in Fig.~\ref{fig:two_link_arm_prioritized_sampling}. We can see that the learned controller using prioritized sampling converges faster while not affecting performance.

%% file: conclusion.tex
\section{Conclusion}
In this paper, we proposed an algorithmic approach to learning a CBF in a data-efficient fashion, starting from an HCBF. By utilizing prioritized sampling, our proposed method is able to find an estimation of the true safe set while using a smaller amount of data compared to uniform sampling. We demonstrated the effectiveness of our approach on a unicycle and a two-link arm system. In future works, we plan to study the integration of a learned performance controller into the proposed framework as well as perform experimental studies on robotic systems testbeds.